\documentclass[a4paper, 12pt]{birkmultarxiv}
\usepackage{amsmath}
\usepackage{amssymb}
\usepackage{amsthm}
\newtheorem{thm}{Theorem}[section]
\newtheorem{lem}[thm]{Lemma}

\numberwithin{equation}{section}
\newcommand{\sgn}{\text{sgn}}

\begin{document}
\title[A quantum particle source]{Dynamical phase transition for a\\ quantum particle source}
\author[M. Butz and H. Spohn]{Maximilian Butz and Herbert Spohn}
\address{%
Technische Universit\"{a}t M\"{u}nchen\\
Zentrum Mathematik\\
Boltzmannstra\ss{}e 3\\
85748 Garching\\
Germany}
\email{\begin{tabbing}
butz@ma.tum.de\\
spohn@ma.tum.de
\end{tabbing}
}

\begin{abstract}
We analyze the time evolution describing a quantum source for noninteracting particles, either bosons or fermions. The growth behaviour of the particle number (trace of the density matrix) is investigated, leading to spectral criteria for sublinear or linear growth in the fermionic case, but also establishing the possibility of exponential growth for bosons. We further study the local convergence of the density matrix in the long time limit and prove the semiclassical limit.
\end{abstract}
\maketitle

\section{Introduction}
Particle sources are an indispensable part of any scattering experiment. Nevertheless in the theoretical description they are mostly disregarded on the basis that a ``suitable'' wave function has been prepared. Of course, a fully realistic modelling of a particle source will be difficult and possibly of marginal interest. However on intermediate grounds, having a simple model source could be of use. The purpose of our paper is to study a, in a certain sense, minimal model. Surprisingly enough, at least to us, we will find that for bosons there is a dynamical phase transition.\par
On the classical level a particle source is easily modeled and, in variation, used widely without further questioning. To explain the principle let us discuss particles in one dimension with position $x_j\in\mathbb{R}$ and velocity $v_j\in\mathbb{R}$, $j=1,2,...$. The source is located at the origin and turned on at time $t=0$. Particles are created at times $0<t_1<t_2<...$ . Once the $j$-th particle is created at time $t_j$, it moves freely as $x_j(t)=v_j(t-t_j)$, $t\geq t_j$. To have a Markov process it is assumed that $(t_{j+1}-t_j)$ are independent and exponentially distributed with rate $\lambda$. Also at the moment of creation the velocity is distributed according to $h(v)\mathrm{d}v$ independently of all the other particles. The average density $f$ on the one-particle phase space is then governed by
\begin{equation}
\label{classsource}
 \frac{\partial}{\partial t}f(x,v,t)+v\frac{\partial}{\partial x}f(x,v,t)=\lambda\delta(x)h(v).
\end{equation}
Of course, one could imagine some statistical dependence. But then the simplicity of equation (\ref{classsource}) is lost. In applications, the left hand side of the transport equation may contain further items as an external potential, a nonlinear collision operator, and the like, to which the source term on the right hand side is simply added. The underlying reasoning for the source term still follows from the statistical assumptions stated above.\\
\indent{}On the basis of (\ref{classsource}) one concludes that, provided $f(x,v,0)=0$, the particle number increases linearly in time,
\begin{equation}
 \int_{\mathbb{R}^2}f(x,v,t)\mathrm{d}x\hspace{0.5mm}\mathrm{d}v=t\lambda\int_{\mathbb{R}}h(v)\mathrm{d}v.
\end{equation}
One can also show, that $f(\cdot,t)$ converges to a steady state as $t\rightarrow\infty$.\par
In the quantum case the Markov process on the many particle level is replaced by a quantum dynamical semigroup on Fock space. As in (\ref{classsource}), the crucial constraint comes from the condition to have a closed equation on the one-particle level, a condition which essentially determines the model uniquely. Particles are created in the pure state $\phi\in\mathcal{H}=L^2\left(\mathbb{R}^d\right)$. Once created they move in $\mathbb{R}^d$ according to the Schr\"{o}dinger equation
\begin{equation}
 i\frac{d}{dt}\psi=H\psi.
\end{equation}
The one-particle Hamiltonian $H$ is a self-adjoint operator with domain $\mathcal{D}(H)\subset\mathcal{H}$. For later purpose, $P_\mathrm{pp}$, resp. $P_\mathrm{ac}$, is the spectral projection of $H$ onto the pure point, resp. absolutely continuous, part of the spectrum. Primarily, we think of the free Schr\"{o}dinger evolution, in which case $H=-\Delta$ (we use units such that the Planck constant $\hbar$ equals $1$ and the particle mass $m$ equals $1/2$). But some of our results also hold abstractly, in particular for the Schr\"{o}dinger operator with a potential, $H=-\Delta+V(x)$. We use the shorthand $H_0=-\Delta$. Our techniques generalize without great efforts to the case of particles being created in a mixed state. Hence we decided to stay with a pure state as minimal model.\\
\indent{}To be more precise, we introduce the Fock space
\begin{equation}
 \mathcal{F}_\pm=\bigoplus_{n=0}^{\infty}S_{\pm}\mathcal{H}^{\otimes n},
\end{equation}
where $S_{\pm}\mathcal{H}^{\otimes n}$ denotes the (anti)symmetrized $n$-fold tensor product of $\mathcal{H}$ with itself, i.e. the either bosonic ($+$) or fermionic ($-$) $n$-particle subspace. The bosonic or fermionic creation and annihilation operators are denoted by $a_\pm(\phi)^*$ and $a_\pm(\phi)$. They satisfy the (anti)commutation relations
\begin{equation}
\label{ccr}
\begin{split}
 [a_\pm(f), a^*_\pm(g)]_\mp&=\left\langle f,g\right\rangle \operatorname{\mathbf{1}}_{\mathcal{F}_\pm},\\
[a_\pm(f), a_\pm(g)]_\mp&=0=[a^*_\pm(f), a^*_\pm(g)]_\mp
\end{split}
\end{equation}
with $\left\langle \cdot,\cdot \right\rangle $ the scalar product of $\mathcal{H}$ and $[A,B]_-=[A,B]=AB-BA$, $[A,B]_+=\{A,B\}=AB+BA$ for operators $A, B$ on $\mathcal{F}_\pm$. We postulate an evolution equation of Lindblad type
\begin{equation}
 \frac{d}{dt}A(t)=\mathcal{L}A(t)
\end{equation}
with bounded $A(t)$ on $\mathcal{F}_\pm$, where $\mathcal{L}$ is the generator of a completely positive dynamical semigroup. To have a closed equation on the one-particle space, $\mathcal{L}$ has to be quadratic in $a_\pm(\phi)^*$, $a_\pm(\phi)$, in other words $\mathcal{L}$ has to generate a quasifree dynamical semigroup. Clearly $\mathcal{L}$ is the sum of the hamiltonian part $\mathcal{L}_0$ and the source part $\mathcal{L}_\mathrm{s}$. $\mathcal{L}_0$ is obviously quasifree and completely positive. If particles are created in the state $\phi\in\mathcal{H}$, $\|\phi\|=1$ with rate $|\lambda|\geq0$, then the only possible source term has to be of the form
\begin{equation}
\label{sourceterm}
 \mathcal{L}_\mathrm{s}A=|\lambda|\big(2a_\pm(\phi)Aa_\pm(\phi)^*-a_\pm(\phi)a_\pm(\phi)^*A-Aa_\pm(\phi)a_\pm(\phi)^* \big)
\end{equation}
Let $\omega_0$  be the initial state as density matrix  on $\mathcal{F}_{\pm}$. Then $\omega_t(A)=\omega_0(e^
{\mathcal{L}t}A)$ defines the state at time $t$. Its one-particle density matrix is given by
\begin{equation}\label{one-point}
\langle g,\rho(t)f \rangle = \omega_t(a^*(f)a(g)).
\end{equation}
Here $\rho(t)^*=\rho(t)$, $\rho(t)\geq0$, $\operatorname{tr}\left(\rho(t) \right)<\infty$ and $\rho(t)\leq\mathbf{1}$ in the case of fermions. If $\omega_0$ is quasifree, then $\omega_t$ is also quasifree. In particular, $\omega_t$
is uniquely determined by $\rho(t)$. 

From (\ref{sourceterm}), (\ref{one-point}) one readily obtains the evolution equation for the one-particle density matrix $\rho(t)$,
\begin{equation}
 \label{timeevolb}
\frac{d}{dt}\rho(t)=-i[H,\rho(t)]+2|\lambda| P_\phi+\lambda\big(P_\phi\rho(t)+\rho(t)P_\phi\big),
\end{equation}
 where $P_\phi=|\phi\left\rangle \right\langle \phi|$ is the orthogonal projection onto $\phi$. We regard $\lambda$ as real parameter, $\lambda\in\mathbb{R}$. Then $\lambda>0$ in (\ref{timeevolb}) is the evolution equation for bosons, while $\lambda<0$ refers to fermions, and we maintain this convention throughout. 
(\ref{timeevolb}) holds for $t\geq0$ and is supplemented by the initial $\rho(0)=\rho_0$. It is the quantum analogue of the classical equation (\ref{classsource}).\\
\indent{}The goal of our paper is a detailed analyis of Equation (\ref{timeevolb}). 

While precise conditions will be given in the main part, let us explain already now the rough overall picture which emerges from our study. We take $H = H_0 = - \Delta$, a  ``reasonable" wave function
$\phi$, and start with $\rho(0) = 0$,  which corresponds to the Fock vaccuum. Hence $\omega_t$
is quasifree and gauge invariant. Basically there is a competition between the speed of transport through $-\Delta$ and the rate $|\lambda|$ at which new particles are supplied. For small $|\lambda|$
the Laplacian dominates. At time $t$ the front particles have travelled a distance of order $t$ from the origin. In essence, the state $\omega_t$ is an incoherent mixture of single particle wave functions,  similar to the classical set-up described by (1.1). If $\lambda \to -\infty$, \textit{i.e.} per unit time a large number of fermions are created, then  the replenishment becomes constrained because of the exclusion principle and, while the number of particles still increases linearly in time, the particle current should level off. In fact, we will show that the current vanishes in the limit
of large production rate.

On the other side for bosons, beyond some critical value $\lambda_\mathrm{c}$, the operator $iH_0 +\lambda P_\phi$ attains an isolated, nondegenerate eigenvalue satisfying $(iH_0 + \lambda P_\phi)\phi_\lambda
= \alpha(\lambda)\phi_\lambda$, $\phi_\lambda \in \mathcal{H}$, with $ \operatorname{Re} \alpha(\lambda)> 0$. Therefore in
the long time limit
\begin{equation}\label{longtime}
\rho(t) \cong P_{\phi_\lambda}\exp[2\operatorname{Re} \alpha(\lambda)t].
\end{equation}
The particle number increases exponentially and a pure Bose condensate  with condensate wave function $\phi_\lambda$ is generated. $\phi_\lambda\to\phi$ as $\lambda \to \infty$. We will
prove this scenario for sufficiently large $\lambda$. The critical regime, which in principle could be more complicated, remains to be explored.

Quasifree dynamical semigroups were introduced in \cite{davies,damoen}. A very readable account is the review by Alicki in \cite{alicki}. He writes down Equations  (\ref{sourceterm}) and  (\ref{timeevolb}) for the case of a general source and sink. He also discusses coupled quantum fields, for which such kind of equations arise in a weak coupling limit. The case of a sink only is considered in \cite{fannes}. A discrete time model is studied, which however in a continuous time limit converges to a quasifree dynamical semigroup of the type (\ref{sourceterm}) with creation and annihilation operators interchanged. We also refer to \cite{attal}, where more recent mathematical contributions are listed. \\
\indent{}To give a brief outline: In the following two sections we properly define the solution to (\ref{timeevolb}) and list the main results. In Sections \ref{sectioninfgrowth} to \ref{sectionexp} we study the particle number, $N(t)=\operatorname{tr}\left(\rho(t) \right)$, when starting with an empty space $\rho_0=0$. In particular we establish both asymptotically linear and exponential growth depending on the parameters. In Section \ref{sectionrestrict} the convergence of the local density matrix as $t\rightarrow\infty$ is investigated. Finally we prove that in the semiclassical limit, $|\lambda|=\mathcal{O}(\epsilon)$ and time, space $\mathcal{O}(\epsilon^{-1})$ the solution of equation (\ref{timeevolb}) converges to the solution of the classical source equation (\ref{classsource}) with $h(v)\sim|\hat\phi(v)|^2$.

\section{Existence of solutions}
The formal solution of equation (\ref{timeevolb}) for $t\geq0$ reads
\begin{equation}
 \label{formal}
\rho(t)=\sgn(\lambda)\left(\mathrm{e}^{(-iH+\lambda P_\phi)t}\mathrm{e}^{(iH+\lambda P_\phi)t}-\mathbf{1}\right)+\mathrm{e}^{(-iH+\lambda P_\phi)t}\rho_0\mathrm{e}^{(iH+\lambda P_\phi)t},
\end{equation}
where $\mathbf{1}=\mathbf{1}_\mathcal{H}$ is the identity map on $\mathcal{H}$.
Since $H$ is selfadjoint, $iH$ generates a strongly continuous unitary group on $\mathcal{H}$ and, considering $\pm\lambda P_\phi$ as a bounded perturbation of this generator, we can apply Theorem 2.1 of \cite{kato}, Chapter IX., \S2 to conclude that $T=iH\pm\lambda P_\phi:\mathcal{D}(H)\rightarrow\mathcal{H}$ still is the generator of a strongly continuous group which is norm-bounded as $\Vert \mathrm{e}^{tT}\Vert\leq \mathrm{e}^{\vert\lambda t\vert}$. Therefore, the (semi)groups occuring in (\ref{formal}) are well-defined. Furthermore, the well-known formula
\begin{equation}
 \mathrm{e}^{(\pm iH+\lambda P_\phi)t}=\mathrm{s}-\lim_{n\rightarrow\infty}\left(\mathbf{1}-\frac{(\pm iH+\lambda P_\phi)t}{n} \right)^{-n}
\end{equation}
implies
\begin{equation}
\label{adj}
 \mathrm{e}^{(-iH+\lambda P_\phi)t}=\left( \mathrm{e}^{(iH+\lambda P_\phi)t}\right) ^*
\end{equation}
for all $t\in\mathbb{R}$.\\ 
\indent{}Next we have to ensure that $\rho(t)$ is actually a density matrix, i.e. a positive trace class operator for all $t\geq0$. From (\ref{formal}) one obtains
\begin{equation}
 \rho(t)=2|\lambda|\int_0^t \mathrm{e}^{(-iH+\lambda P_\phi)s} P_\phi \mathrm{e}^{(iH+\lambda P_\phi)s} ds+\mathrm{e}^{(-iH+\lambda P_\phi)t}\rho_0\mathrm{e}^{(iH+\lambda P_\phi)t}
\end{equation}
in which the integral is a trace-class valued Riemann integral (recall that the semigroups are strongly continuous). Since $\rho_0$ and $P_\phi$ are positive trace class operators, (\ref{adj}) implies the same property for $\rho(t)$. In case $\lambda<0$, we have to check the fermionic property $\rho(t)\leq\mathbf{1}$. But if the system starts in a fermionic state, $\rho_0\leq\mathbf{1}$, (\ref{formal}) can be reordered to 
\begin{equation}
 \mathbf{1}-\rho(t)=\mathrm{e}^{(-iH+\lambda P_\phi)t}\left(\mathbf{1} -\rho_0\right) \mathrm{e}^{(iH+\lambda P_\phi)t},
\end{equation}
which stays positive for all $t$. Thus the fermionic property of $\rho_0$ is preserved in time.\

\section{Main results}
\subsection{Asymptotics of the particle number}
Most basically one would like to know the number of particles, $N(t)$, produced by the source. Na\"{\i}vely one would expect $N(t)$ to grow linearly. However, the statistics of the particles induces an effective attraction, respectivly repulsion, which might change such simplistic picture. From (\ref{formal}) one easily calculates
\begin{equation}
\label{trg}
 \frac{d}{dt}\operatorname{tr}(\rho(t))= 2|\lambda|\Vert \mathrm{e}^{(-iH+\lambda P_\phi)t} \phi\Vert^2+2\lambda\left\langle\phi,\mathrm{e}^{(-iH+\lambda P_\phi)t}\rho_0\mathrm{e}^{(iH+\lambda P_\phi)t}\phi\right\rangle.
\end{equation}
Together with the observation that for all $\psi\in\mathcal{H}$, 
\begin{equation}
\label{mono}
 \frac{d}{dt}\Vert \mathrm{e}^{(-iH+\lambda P_\phi)t} \psi\Vert^2=2\lambda\left|\left\langle \phi,\mathrm{e}^{(-iH+\lambda P_\phi)t} \psi\right\rangle \right|^2,
\end{equation}
equation (\ref{trg}) immediately implies that for any Hamiltonian $H$, the number of particles grows at least linearly in time for $\lambda>0$, and at most linearly for $\lambda<0$. The number of fermions is monotonically increasing, which is an easy consequence of (\ref{trg}) together with the fermionic property $\rho_0\leq \mathbf{1}$ and the fact that $\Vert \mathrm{e}^{(-iH+\lambda P_\phi)t} \phi\Vert^2=\Vert \mathrm{e}^{(+iH+\lambda P_\phi)t} \phi\Vert^2$ (their derivatives with respect to $t$ are equal).\\
\indent{}Concerning the asymptotic growth of $N(t)$, the initial density $\rho_0$ does not change the qualitative behaviour, and is therefore set to $\rho_0=0$. In addition to the upper bound $N(t)\leq2|\lambda| t$ for the number of fermions,  there is an easy characterization of those source states for which the number of fermions stays bounded.
\begin{thm}
\label{infgrowth}
 Let $\lambda<0$. $N(t)$ stays bounded as $t\rightarrow\infty$ if and only if the source state $\phi$ is a finite linear combination of eigenvectors of $H$. In this case
\begin{equation}
 \lim_{t\rightarrow\infty}N(t)=N.
\end{equation}
with $N$ the number of different eigenvalues corresponding to these eigenvectors.

\end{thm}
\noindent{}Concerning the distinction of source states generating linear or sublinear growth of $N(t)$ in the fermionic case we have the following result.
\begin{thm}
 \label{linsublin}
Let $\lambda<0$. If $P_{\mathrm{pp}}\phi=\phi$, then the limit rate of particle production vanishes and $N(t)$ grows sublinearly. On the other hand, if $P_\mathrm{ac}\phi\neq0$, then $N(t)$ increases linearly. 
\end{thm}
In particular, for $H=H_0$, and arbitrary $\phi$, the number of fermions increases linearly in time.\\
\indent{}For bosons, linear growth is not an upper, but a lower bound for $N(t)$. But also in this case one can characterize explicitely a class of source states which yield linear growth for small $\lambda$. 
\begin{thm}
\label{lingrowth}
 Let $\phi\in\mathcal{H}$, $\|\phi\|=1$ and assume that $\tau:=\int_0^\infty|\left\langle\phi,\mathrm{e}^{-iHt}\phi \right\rangle|\mathrm{d}t<\infty$. Then for all $\lambda<0$ or $0<\lambda<\tau^{-1}$,  $\frac{d}{dt}N(t)$ converges to a non-zero limit as $t\rightarrow\infty$, which satisfies
\begin{equation}
\label{growthrates}
\begin{split}
 2|\lambda|&\leq\lim_{t\rightarrow\infty}\frac{d}{dt}N(t)\leq\frac{2|\lambda|}{\left( {1-\lambda\tau}\right) ^2}\hspace{5mm}(\lambda>0),\\
\frac{2|\lambda|}{\left( {1-\lambda\tau}\right) ^2}&\leq\lim_{t\rightarrow\infty}\frac{d}{dt}N(t)\leq2|\lambda|\hspace{5mm}(\lambda<0).
\end{split}
\end{equation}
\end{thm}
 \noindent{}The quantity $\left\langle\phi,\mathrm{e}^{-iHt}\phi \right\rangle$ is the overlap between the source state at time $0$ and at time $t$. Thus $\tau$ should be regarded as a measure for how long it takes the time evolution to transport an emitted particle away from its source. In this context, $|\lambda|<\tau^{-1}$ means, that the strength of the source is smaller than the ``transport capacity'' of the time evolution. Therefore, the emitted particles hardly influence each other, and in essence the bosonic or fermionic character does not show in the evolution of the density matrix. The limit particle production is constant as in the classical case. An example for such $\phi$  in the physically relevant case $\mathcal{H}=L^2\left(\mathbb{R}^3 \right)$, $H_0=-\Delta$ is $\phi\in L^1\cap L^2\left(\mathbb{R}^3 \right)$, since $\left\langle\phi,\mathrm{e}^{-iHt}\phi \right\rangle$ has a $t^{-3/2}$ decay, as can be deduced from the free propagator
\begin{equation}
\label{kernel}
 \left( \mathrm{e}^{-iH_0t}\phi\right) (x)=\left(4\pi i t \right)^{-\frac{d}{2}}\int_{\mathbb{R}^d}e^{\frac{i\vert x-y\vert^2}{4t}}\phi(y)\mathrm{d}y 
\end{equation}
for $\phi\in L^1\cap L^2\left(\mathbb{R}^d \right)$ (compare \cite{resi2}, Chapter IX.7).
For fermions, this is also an example where one has an explicit estimate for the speed of particle production. However, the form of the lower bound in (\ref{growthrates}) already suggests that the limit growth rate might decrease to $0$ as $\lambda\rightarrow-\infty$. In fact, we have
\begin{thm}
\label{ratetozero}
 For $\phi\in\mathcal{D}(H)$, we have
\begin{equation}
 \lim_{\lambda\rightarrow-\infty}\lim_{t\rightarrow\infty}\left(\frac{d}{dt}N_\lambda(t)\right)=0.
\end{equation}
\end{thm}
 For bosons, the behavior at large values of $\lambda$ is qualitatively very different. In this regime, exponential growth occurs for arbitrary choices of the Hamiltonian $H$ and the source state $\phi$.
\begin{thm}
 \label{pointspec}
For sufficiently large $\lambda>0$ the operator $iH+\lambda P_\phi$ has an eigenvalue $\alpha(\lambda)$ with positive real part, and $\lim_{\lambda\rightarrow\infty}\lambda^{-1}\alpha(\lambda)=1$.
\end{thm}
If such an eigenvalue $\alpha(\lambda)$  with normalized eigenvector $\psi$ exists, the number of bosons can be estimated by 
\begin{equation}
\label{expest}
 \operatorname{tr}(\rho(t))\geq\left\langle \psi, \rho(t) \psi\right\rangle\geq {e}^{2\operatorname{Re} \alpha(\lambda)t}-1.
\end{equation}
Hence Theorem \ref{pointspec} implies the existence of a critical strength $\lambda_\mathrm{c}\geq0$ such that the number of particles grows exponentially in time for all $\lambda>\lambda_\mathrm{c}$. For source states $\phi$ as in Theorem \ref{lingrowth}, our result implies a dynamical phase transition from linear ($0<\lambda<\tau^{-1}$) to exponential ($\lambda>\lambda_\mathrm{c}$) growth. The general picture is more complicated, however. If $H$ has point spectrum, it is obvious that choosing $\phi$ as an eigenvector of $H$ will generate exponential growth for all $\lambda>0$: $N(t)={e}^{2\lambda t}-1$. But also for $H=H_0$ there  are source states $\phi$ which generate exponential growth for \emph{all} $\lambda>0$.
\begin{thm}
\label{alllambda}
 Let $\phi\in\mathcal{D}(H_0)=H^2\left(\mathbb{R}^d\right) $ with Fourier transform
\begin{equation}
\label{sourcestate}
 \hat\phi(p)=|S^{d-1}|^{-{1}/{2}}|p|^{(1-d)/{2}}\sqrt{\frac{3}{\pi(|p|^6+1)}}.
\end{equation}
Then $iH_0+\lambda P_\phi$ has an eigenvalue $\alpha(\lambda)$ with positive real part for all $\lambda>0$.
\end{thm}
In the proof of this theorem, the eigenvalue will not be computed explicitly. With a view towards the semiclassical limit, one can infer that $\operatorname{Re}\alpha(\lambda)=o(\lambda)$ has to hold. In this limit one considers a source with activity $\lambda=c\epsilon$ on a time scale $\epsilon^{-1}t$ and thus the exponent in (\ref{expest}),
\begin{equation}
 2\operatorname{Re}\alpha(c\epsilon)\epsilon^{-1}t,
\end{equation}
has to vanish in the limit $\epsilon\rightarrow0$ so to yield the linear classical growth behaviour.

\indent{}In contrast to Theorem \ref{lingrowth} the source state (\ref{sourcestate}) has an overlap decay as $1/\sqrt{t}$, thus $\tau=\infty$. As a consequence, the emitted particles stay close to the source for a long time and, being bosonic, they pull further particles out of the source, which leads to an exponential growth of the particle number regardless of how small $\lambda$.

\subsection{Convergence of $\rho(t)$}
Having studied $N(t)$ one may ask whether $\rho(t)$ has a limit as $t\rightarrow\infty$. Since mostly $N(t)\rightarrow\infty$ for large $t$, the natural notion is to study the local limit of $\rho(t)$. Let $\Omega\subset\mathbb{R}^d$ be a bounded region and let $P_\Omega$ denote the orthogonal projection of $\mathcal{H}$ onto the subspace $L^2(\Omega)$. We consider the number of particles in $\Omega$, i.e.
\begin{equation}
 N_\Omega(t)=\operatorname{tr}(P_\Omega\rho(t)P_\Omega).
\end{equation}
If this quantity stays bounded as $t\rightarrow\infty$, one can use that $P_\Omega\rho(t)P_\Omega$ has a positive time derivative (at least for $\rho_0=0$) to infer that the restricted density matrix $P_\Omega\rho(t)P_\Omega$ has even a trace class limit.

A first result shows that for fermions and $H=H_0$ this limit, as far as it exists, does not depend on the initial condition $\rho_0$.
\begin{thm}
\label{independence}
 For $\Omega\subset\mathbb{R}^d$, with finite Lebesgue measure $|\Omega|<\infty$, $\lambda<0$ and $\rho_0$ an arbitrary trace class operator, it holds
\begin{equation}
 \Vert P_\Omega \mathrm{e}^{(-iH_0+\lambda P_\phi)t}\rho_0\mathrm{e}^{(iH_0+\lambda P_\phi)t} P_\Omega\Vert_{\operatorname{tr}}\rightarrow0\hspace{4mm}\mbox{as }t\rightarrow\infty.
\end{equation}
\end{thm}
For all remaining results, we return to $\rho_0=0$ again. The next result is rather obvious: If particles are generated at most linearly in time, and if they move away from  a certain region $\Omega$ in space within finite time, the number of particles in $\Omega$ will not diverge:
\begin{thm}
\label{finlim}
 Let us choose $\phi$ and $\lambda$ such that $N(t)$ has a linear bound and assume that
\begin{equation}
\label{L1}
 \int_0^\infty\|P_\Omega \mathrm{e}^{-iHt}\phi\|\mathrm{d}t<\infty.
\end{equation}
 Then $\operatorname{tr}(P_\Omega\rho(t)P_\Omega)$ approaches a finite limit as $t\rightarrow\infty$.
\end{thm}
In particular, for $H=H_0$ this theorem applies to $\Omega$ with $|\Omega|<\infty$ for all $\phi\in L^1\cap L^2\left(\mathbb{R}^3 \right)$ and all $\lambda<\tau^{-1}$: (Sub)linearity follows from Theorem \ref{lingrowth}, and (\ref{L1}) is deduced from the explicit form of the integral kernel (\ref{kernel}).\\
\indent{}The next result is characteristic for the behaviour of a fermionic particle source. Consider the fact that in a region $\Omega$ of finite measure, there are only finitely many states with kinetic energy below a certain bound (cf. \cite{perry}, p.\ 27). If the source state has no high-energy contributions, it should only be able to charge a finite number of states, so that the particle number in $\Omega$ stays finite.
\begin{thm}
\label{fermicomp}
 Let $H=H_0$, $\phi\in\mathcal{H}$, with Fourier transform $\hat{\phi}$ supported in $K\subset\mathbb{R}^d$ where $|K|<\infty$, and let $\Omega\subset\mathbb{R}^d$ with $|\Omega|<\infty$. Then, for $\lambda<0$,
\begin{equation}
 \lim_{t\rightarrow\infty}\operatorname{tr}(P_\Omega\rho(t)P_\Omega)\leq(2\pi)^{-d}|\Omega||K|<\infty.
\end{equation}
\end{thm}
The upper bound corresponds to the phase space volume of fermions, i. e. to one fermion per unit cell.
\subsection{Semiclassical limit}
To pass from the Hilbert space formulation of quantum mechanics to phase space, one convenient tool is the Wigner transform. One defines
\begin{equation}
 \label{wigner}
W[\kappa](x,p)=(2\pi)^{-d}\int_{\mathbb{R}^d}\mathrm{e}^{-i{p\cdot y}}\kappa\left(x+\frac{y}{2},x-\frac{y}{2} \right) \mathrm{d}y,
\end{equation}
where $\kappa\in L^2\left( \mathbb{R}^d_x\times\mathbb{R}^d_y\right)$ is the integral kernel of a Hilbert-Schmidt operator.
In particular, since all trace-class operators are Hilbert-Schmidt, we can consider the Wigner transform for any density matrix.
Denoting, as usual, the semiclassical parameter by $\epsilon$, the semiclassical limit corresponds to the small $\epsilon$ behaviour of
\begin{equation}
\label{deffeps}
f^\epsilon(X,P,T)=\epsilon^{-d}W\left[ \rho^\epsilon\left(\epsilon^{-1}T\right)\right] \left(\epsilon^{-1}X,P \right), \hspace{7mm}(X,P,T)\in\mathbb{R}^d\times\mathbb{R}^d\times\mathbb{R}^+,
\end{equation}
where $\rho^\epsilon(t)$ is the solution of
\begin{equation}
 \label{microevol}
 \frac{d}{dt}\rho^\epsilon(t)=-i\left[ H_0,\rho^\epsilon(t)\right]+2|c|\epsilon P_\phi
+ c\epsilon\left( P_\phi\rho^\epsilon(t)+\rho^\epsilon(t)P_\phi\right)
\end{equation}
for all $t\geq0$, with the convention $H_0=-\Delta/2$. The initial densities $\rho_0^\epsilon$ are chosen such that $\|\rho_0^\epsilon\|_{\operatorname{tr}}$ is bounded uniformly in $\epsilon>0$ and
\begin{equation}
\label{indist}
 \epsilon^{-d}W\left[ \rho^\epsilon_0\right] \left(\epsilon^{-1}X,P \right)\rightarrow g(X,P)
\end{equation}
for some distribution $g\in\mathcal{D}'\left(\mathbb{R}^{d}_X\times\mathbb{R}^{d}_P\right)$. The equations (\ref{deffeps})-(\ref{indist}) describe a quantum particle source producing particles on a microscopic scale $(x,p,t)$ correlated to the macroscopic scale $(X,P,T)$ by $(X,P,T)=(\epsilon x, p, \epsilon t)$, the standard semiclassical scaling, see e.g. \cite{nier}. To have a bounded rate $|c|$ on the macroscopic time scale, we have set $\lambda=c\epsilon$. 
In this case, we have the following convergence result for the phase space density $f^\epsilon$.
\begin{thm}
\label{limitdist}
 For all $g\in\mathcal{D}'\left(\mathbb{R}^{d}_x\times\mathbb{R}^{d}_p\right) $, $\phi\in\mathcal{H}$, $\Vert\phi\Vert=1$, and for all $T\geq0$, the limit
\begin{equation}
 \lim_{\epsilon\rightarrow0}f^\epsilon(X,P,T)=f^0(X,P,T)=g(X-PT,P)+2|c|\int_0^T\delta(X-Ps)\vert\hat\phi(P)\vert^2\mathrm{d}s
\end{equation}
holds in the topology of $\mathcal{D}'\left(\mathbb{R}^{d}_X\times\mathbb{R}^{d}_P\right)$. This limit solves
\begin{equation}
\label{limitdiff}
\begin{split}
 \frac{\partial}{\partial T}f^0(X,P,T)+P\cdot\nabla_Xf^0(X,P,T)&=2|c|\delta(X)\vert\hat\phi(P)\vert^2,\\
f^0(X,P,0)&=g(X,P)
\end{split}
\end{equation}
in the sense of distributions on phase space.
\end{thm}
(\ref{limitdiff}) is the weak form of equation (\ref{classsource}) with a source term defined through the semiclassical limit of $\phi$.

A natural step would be to include an external potential varying on the macroscopic scale. We leave this as an open problem.
\section{Particle production in the fermionic case}
\label{sectioninfgrowth}
\subsection{Proof of Theorem \ref{infgrowth}}
\begin{proof}
Since we are in the case $\rho_0=0$, $\rho(t)$ is differentiable in trace class with a positive operator as derivative,
\begin{equation}
 \frac{d}{dt}\rho(t)= 2|\lambda| \mathrm{e}^{(-iH+\lambda P_\phi)t} P_\phi \mathrm{e}^{(iH+\lambda P_\phi)t}.
\end{equation}
Thus $\|\rho(t)-\rho(s)\|_{\operatorname{tr}}=|\operatorname{tr}(\rho(t))-\operatorname{tr}(\rho(s))|$, so if the particle number stays bounded and $\operatorname{tr}\left(\rho(t)\right)$ approaches a finite limit from below, $\rho(t)$ converges in trace class as $t$ tends to infinity. As a fermionic density matrix, the limit $\rho_\infty$ obeys $0\leq\rho_\infty\leq \mathbf{1}$ and it stays invariant under the time evolution given by (\ref{formal}),
\begin{equation}
\label{invariance}
 \rho_\infty=\mathbf{1}-\mathrm{e}^{(-iH+\lambda P_\phi)t}\mathrm{e}^{(iH+\lambda P_\phi)t}+\mathrm{e}^{(-iH+\lambda P_\phi)t}\rho_\infty \mathrm{e}^{(iH+\lambda P_\phi)t}.
\end{equation}
 Furthermore, since $\operatorname{tr}(\rho(t))$ converges to a finite limit, its (monotonously decreasing) time derivative tends to zero, and therefore,
\begin{equation}
\begin{split}
 \left\langle\phi,(\mathbf{1}-\rho_\infty)\phi \right\rangle&=\lim_{t\rightarrow\infty}\left\langle\phi,(\mathbf{1}-\rho(t))\phi \right\rangle=\lim_{t\rightarrow\infty}\Vert \mathrm{e}^{(iH+\lambda P_\phi)t} \phi\Vert^2\\
&=\lim_{t\rightarrow\infty}\Vert \mathrm{e}^{(-iH+\lambda P_\phi)t} \phi\Vert^2=\lim_{t\rightarrow\infty}(2|\lambda|)^{-1} \frac{d}{dt}\operatorname{tr}(\rho(t))=0.
\end{split}
\end{equation}

\noindent{}Since $\mathbf{1}-\rho_\infty$ is positive, $0=\left\langle\phi,(\mathbf{1}-\rho_\infty)\phi \right\rangle$ can only hold if $(\mathbf{1}-\rho_\infty)\phi=0$. Together with (\ref{invariance}) this yields:
\begin{equation}
\begin{split}
 \frac{d}{dt}\left(\mathrm{e}^{iHt}\rho_\infty \mathrm{e}^{-iHt} \right)&=|\lambda|\mathrm{e}^{iHt}P_\phi \mathrm{e}^{(-iH+\lambda P_\phi)t}(\mathbf{1}-\rho_\infty)\mathrm{e}^{(iH+\lambda P_\phi)t}\mathrm{e}^{-iHt} \\
&\hspace{4mm}+|\lambda|\mathrm{e}^{iHt} \mathrm{e}^{(-iH+\lambda P_\phi)t}(\mathbf{1}-\rho_\infty) \mathrm{e}^{(iH+\lambda P_\phi)t}P_\phi \mathrm{e}^{-iHt}\\
&=|\lambda| \mathrm{e}^{iHt}\left(P_\phi (\mathbf{1}-\rho_\infty)+ (\mathbf{1}-\rho_\infty)P_\phi \right) \mathrm{e}^{-iHt}=0.
\end{split}
\end{equation}
Therefore, we have for $t\in\mathbb{R}$
\begin{equation}
\label{similar}
 \mathrm{e}^{iHt}\rho_\infty \mathrm{e}^{-iHt}=\rho_\infty.
\end{equation}
Now consider the eigenspace $V$ of $\rho_\infty$ corresponding to the eigenvalue $1$. $V$ has a positive, but finite dimension because $\phi\in V$ and $\rho_\infty$ is trace class, and due to (\ref{similar}), $V$ is invariant under the action of the group $\mathrm{e}^{-iHt}$. So since $V$ is finite-dimensional, we have $V\subset\mathcal{D}(H)$ and $V$ has a orthonormal basis of eigenvectors of $H$, and $\phi$ can be written as a linear combination of those.

Conversely, assume that $\phi=\sum_{n=1}^Nc_n\psi_n$ with $\|\psi_n\|=1$, $H\psi_n=\omega_n\psi_n$, with all $c_n\neq0$ and all $\omega_n$ different. $-iH+\lambda P_\phi$ is a bounded operator on $W:=\operatorname{span}(\psi_1,...,\psi_n)$, so $\mathrm{e}^{(-iH+\lambda P_\phi)t}\phi$ can be calculated by the power series, to see that 
\begin{equation}
\rho(t)= 2|\lambda|\int_0^t \mathrm{e}^{(-iH+\lambda P_\phi)s} P_\phi \mathrm{e}^{(iH+\lambda P_\phi)s}\mathrm{d}s.
\end{equation}
contains only states from $W$, which, together with the fermionic property, yields the bound $\operatorname{tr}(\rho(t))\leq N$. Thus $\rho_\infty$ exists and, by (\ref{similar}) commutes with $H$ on $W$, so that they have a common basis of eigenvectors, which necessarily means
\begin{equation}
 \rho_\infty=\sum_{n=1}^Nb_n|\psi_n \left\rangle\right\langle\psi_n|.
\end{equation}
Now it only remains to check that $\rho_\infty=\mathbf{1}_W$, which follows from
\begin{equation}
 \sum_{n=1}^Nc_n\psi_n=\phi=\rho_\infty\phi=\sum_{n=1}^Nc_n b_n\psi_n.
\end{equation}
and the fact that all $c_n\neq0$. So $\operatorname{tr}(\rho_\infty)=\operatorname{dim}(W)=N$.
\end{proof}

\subsection{Proof of Theorem \ref{linsublin}}
\begin{proof}
For a state $\phi\in\mathcal{H}_{\mathrm{pp}}$ we have to show that 
\begin{equation}
 \lim_{t\rightarrow\infty}\frac{d}{dt}\operatorname{tr}(\rho(t))= \lim_{t\rightarrow\infty}2|\lambda|\Vert \mathrm{e}^{(-iH+\lambda P_\phi)t} \phi\Vert^2=0.
\end{equation}
To this end, expand $\phi$ in a basis of eigenvectors of $H$,
\begin{equation}
  \phi=\sum_{n=1}^\infty c_n\psi_n
\end{equation}
with $\|\psi_n\|=1$, $H\psi_n=\omega_n\psi_n$. Since 
\begin{equation}
\label{convA}
\begin{split}
 \mathrm{e}^{(iH+\lambda P_\phi)t}\mathrm{e}^{(-iH+\lambda P_\phi)t}&\geq0,\\
\frac{d}{dt}\mathrm{e}^{(iH+\lambda P_\phi)t}\mathrm{e}^{(-iH+\lambda P_\phi)t}&=-2|\lambda|\mathrm{e}^{(iH+\lambda P_\phi)t}P_\phi \mathrm{e}^{(-iH+\lambda P_\phi)t}\leq0,
\end{split}
\end{equation}
polarization implies that $\mathrm{e}^{(iH+\lambda P_\phi)t}\mathrm{e}^{(-iH+\lambda P_\phi)t}$ converges weakly to a bounded, positive, selfadjoint operator $A$ with
\begin{equation}
\label{invarA}
 A=\mathrm{e}^{(iH+\lambda P_\phi)t}A\mathrm{e}^{(-iH+\lambda P_\phi)t}.
\end{equation}
As the $\psi_n$ are eigenvectors of $H$, 
\begin{equation} 
 \left\langle\psi_n,\mathrm{e}^{-iHt}A\mathrm{e}^{iHt}\psi_n \right\rangle= \left\langle\psi_n,A\psi_n \right\rangle.
\end{equation}
The last two equations prove
\begin{equation}
\begin{split}
 0&=\frac{d}{dt}\left\langle\psi_n,\mathrm{e}^{-iHt}A\mathrm{e}^{iHt}\psi_n \right\rangle\\
&=\lambda\left\langle\psi_n,\mathrm{e}^{-iHt}(P_\phi A+AP_\phi) \mathrm{e}^{iHt}\psi_n \right\rangle\\
&=\lambda\left\langle\psi_n,(P_\phi A+AP_\phi)\psi_n \right\rangle.
\end{split}
\end{equation}
Taking the sum over all $n$, one obtains $\left\langle\phi,A \phi\right\rangle=0$, (even $A\phi=0$). By the definition of $A$, this means
\begin{equation}
 \lim_{t\rightarrow0}\Vert \mathrm{e}^{(-iH+\lambda P_\phi)t} \phi\Vert^2=\lim_{t\rightarrow0}\left\langle\phi, \mathrm{e}^{(iH+\lambda P_\phi)t}\mathrm{e}^{(-iH+\lambda P_\phi)t}\phi \right\rangle =\left\langle\phi,A \phi\right\rangle=0,
\end{equation}
which proves the first assertion.\par
For the proof of the second part of the theorem, assume that there is a source state $\phi$ with $P_\mathrm{ac}\phi\neq0$ and a $\lambda<0$ such that the number of particles grows sublinearly in time, i.e.
\begin{equation}
 h(t)=\left\Vert\mathrm{e}^{(-iH+\lambda P_\phi)t}\phi\right\Vert=\left\Vert\mathrm{e}^{(iH+\lambda P_\phi)t}\phi\right\Vert\rightarrow0\hspace{5mm}(t\rightarrow\infty).
\end{equation}
By the same arguments as in (\ref{convA}), we have the weak convergence
\begin{equation}
 \mathrm{e}^{(iH+\lambda P_\phi)t}\mathrm{e}^{(-iH+\lambda P_\phi)t}\rightarrow A 
\end{equation}
with $A$ having the invariance property (\ref{invarA}), and, by assumption,
\begin{equation}
 A\phi=0.
\end{equation}
Following the proof of Theorem \ref{infgrowth}, one can deduce from (\ref{invarA}) and $A\phi=0$ that 
\begin{equation}
\label{invarAunit}
 A=\mathrm{e}^{iHt}A\mathrm{e}^{-iHt}
\end{equation}
for all $t$. Thus $A$ and $H$ commute, and so do $A$ and all operators given by the functional calculus for $H$,
\begin{equation}
\label{commuteA}
 Af(H)=f(H)A,\hspace{5mm}
\end{equation}
for all bounded Borel functions $f$ on $\mathbb{R}$. We will now apply this fact to a certain choice of $f$.
For $\psi=P_\mathrm{ac}\phi\neq0$, we have $\mu_\psi(\mathrm{d}E)=\rho(E)\mathrm{d}E$ with a positive $L^1(\mathbb{R})$ function $\rho$. Defining the Borel function
\begin{equation}
 f_c=\mathbf{1}{\{\rho\leq c\}}\cdot\mathbf{1}\left( {\mathbb{R}\setminus\mathrm{supp}(\mu_\mathrm{pp})}\right) \cdot\mathbf{1}\left( {\mathbb{R}\setminus\mathrm{supp}(\mu_\mathrm{sing})}\right) 
\end{equation}
for $c>0$, we have
\begin{equation}
 \psi_c=f_c(H)\phi\rightarrow\psi\hspace{5mm}(c\rightarrow\infty)
\end{equation}
and
\begin{equation}
 A\psi_c=Af_c(H)\phi=f_c(H)A\phi=0.
\end{equation}
Thus for $c$ fixed large enough, $\psi_c\neq0$ and 
\begin{equation}
\label{tozero}
 \lim_{t\rightarrow\infty}\left\Vert\mathrm{e}^{(-iH+\lambda P_\phi)t}\psi_c\right\Vert^2=\langle\psi_c,A\psi_c\rangle=0.
\end{equation}
Considering the semigroup as a perturbation of the unitary group, we can write
\begin{equation}
 \mathrm{e}^{(-iH+\lambda P_\phi)t}\psi_c=\mathrm{e}^{-iHt}\psi_c+\lambda\int_0^t\mathrm{e}^{(-iH+\lambda P_\phi)(t-s)}\phi\left\langle\phi,\mathrm{e}^{-iHs}\psi_c \right\rangle\mathrm{d}s.
\end{equation}
By the choice of $\psi_c$, the scalar product reads
\begin{equation}
 \left\langle\phi,\mathrm{e}^{-iHs}\psi_c \right\rangle=\int_\mathbb{R}\mathbf{1}{\{\rho\leq c\}}\rho(E)e^{iEs}\mathrm{d}E,
\end{equation}
and therefore is an $L^2(\mathbb{R}_s)$ function with norm not larger than $\sqrt{2\pi c}$. Choosing $\epsilon>0$ and $T>0$ such, that 
\begin{equation}
 \int_T^\infty|\left\langle\phi,\mathrm{e}^{-iHs}\psi_c \right\rangle|^2\mathrm{d}s<\epsilon
\end{equation}
we can split the integral in two parts to see for all $t\geq T$
\begin{equation}
\label{4.28}
\begin{split}
 \left\Vert\mathrm{e}^{(-iH+\lambda P_\phi)t}\psi_c\right\Vert\geq&\left\Vert\psi_c\right\Vert-|\lambda|\int_0^Th(t-s)\left|\left\langle\phi,\mathrm{e}^{-iHs}\psi_c \right\rangle\right|\mathrm{d}s\\&\hspace{5mm}-|\lambda|\sup_{\chi\in\mathcal{H}, \Vert\xi\Vert=1}\int_T^t\left|\langle\chi,\mathrm{e}^{(-iH+\lambda P_\phi)(t-s)}\phi\rangle\left\langle\phi,\mathrm{e}^{-iHs}\psi_c \right\rangle\right|\mathrm{d}s.
\end{split}
\end{equation}
Now actually both scalar products in the second integral are $L^2$ functions, since
\begin{equation}
 \frac{d}{ds} \Vert \mathrm{e}^{(iH+\lambda P_\phi)s}\chi\Vert^2=2\lambda\left|\left\langle\phi, \mathrm{e}^{(iH+\lambda P_\phi)s}\chi\right\rangle\right|^2
\end{equation}
and thus
\begin{equation}
\label{L2est}
 \int_0^\infty\left|\left\langle\phi, \mathrm{e}^{(iH+\lambda P_\phi)s}\chi\right\rangle\right|^2\mathrm{d}s=\frac{\Vert\chi\Vert^2-\lim_{t\rightarrow\infty}\Vert \mathrm{e}^{(iH+\lambda P_\phi)t}\chi\Vert^2}{2|\lambda|}\leq\frac{\Vert\chi\Vert^2}{|2\lambda|}.
\end{equation}
Therefore an application of the Cauchy-Schwarz inequality to (\ref{4.28}) shows
\begin{equation}
\left\Vert\mathrm{e}^{(-iH+\lambda P_\phi)t}\psi_c\right\Vert\geq\left\Vert\psi_c\right\Vert-|\lambda|\int_0^Th(t-s)\left|\left\langle\phi,\mathrm{e}^{-iHs}\psi_c \right\rangle\right|\mathrm{d}s-\sqrt{\frac{|\lambda|\epsilon}{2}}.
\end{equation}
For fixed $T$, the integral in the last line tends to zero as $t\rightarrow\infty$ by dominated convergence and the assumption that $h(t)\rightarrow0$. Since one can take $\epsilon$ arbitrarily small, this would imply
\begin{equation}
 \lim_{t\rightarrow\infty}\left\Vert\mathrm{e}^{(-iH+\lambda P_\phi)t}\psi_c\right\Vert\geq\left\Vert\psi_c\right\Vert>0,
\end{equation}
contradicting (\ref{tozero}). Thus sublinear growth is not possible for source states $\phi$ with $P_\mathrm{ac}\phi\neq0$.
\end{proof}
\section{Linear growth}
\label{sectionlingrowth}
\subsection{Proof of Theorem \ref{lingrowth}}
\begin{proof}
For $t>0$ define $h(t)=\left\Vert\mathrm{e}^{(-iH+\lambda P_\phi)t}\phi\right\Vert$. Considering the nonunitary time evolution as a perturbation of the unitary group, one can write
\begin{equation}
 \mathrm{e}^{(-iH+\lambda P_\phi)t}\phi=\mathrm{e}^{-iHt}\phi+\lambda\int_0^t\mathrm{e}^{(-iH+\lambda P_\phi)(t-s)}\phi\left\langle\phi,\mathrm{e}^{-iHs}\phi \right\rangle\mathrm{d}s.
\end{equation}
For $0<\lambda<\tau^{-1}$, the monotonicity of $h(t)$ implies
\begin{equation}
 \begin{split}
  h(t)&\leq1+\lambda\tau h(t)\\
h(t)&\leq(1-\lambda\tau)^{-1},
 \end{split}
\end{equation}
so $\lim_{t\rightarrow\infty}h(t)$ exists and is bounded by $(1-\lambda\tau)^{-1}$. For $\lambda<0$, $\lim_{t\rightarrow\infty}h(t)=h_\infty$ exists since $h(t)$ decreases monotonically, and dominated convergence implies the estimate
\begin{equation}
 \begin{split}
  h_\infty&\geq1-|\lambda|\tau h_\infty\\
h_\infty&\geq(1-\lambda\tau)^{-1}.
 \end{split}
\end{equation}
Setting $\rho_0=0$ in (\ref{trg}) and using the unitarity of $\mathrm{e}^{iHt}$, this implies the existence of $\lim_{t\rightarrow\infty}\frac{d}{dt}N(t)$ and the estimates
\begin{equation}
 \lim_{t\rightarrow\infty}\frac{d}{dt}N(t)\leq\frac{2|\lambda|}{\left( {1-\lambda\tau}\right) ^2} \hspace{5mm}(\lambda>0)
\end{equation}
and
\begin{equation}
 \lim_{t\rightarrow\infty}\frac{d}{dt}N(t)\geq\frac{2|\lambda|}{\left( {1-\lambda\tau}\right) ^2}\hspace{5mm}(\lambda<0).
\end{equation}
The other bounds follow from (\ref{trg}) and (\ref{mono}).

\end{proof}

\subsection{Proof of Theorem \ref{ratetozero}}
For $\lambda<0$, $\phi\in\mathcal{D}(H)$, one has
\begin{equation}
 \frac{d}{dt}\mathrm{e}^{(iH-\lambda P_\phi)t}\mathrm{e}^{\lambda P_\phi t}\phi=ie^{\lambda t}H\phi,
\end{equation}
and thus
\begin{equation}
 \mathrm{e}^{(-iH+\lambda P_\phi)t}\phi=e^{\lambda t}\phi-i\int_0^te^{\lambda s}\mathrm{e}^{(-iH+\lambda P_\phi)(t-s)}H\phi\mathrm{d}s.
\end{equation}
Therefore,
\begin{equation}
 \left\Vert\mathrm{e}^{(-iH+\lambda P_\phi)t}\phi\right\Vert\leq e^{-|\lambda|t}+\frac{1}{|\lambda|}\rightarrow\frac{1}{|\lambda|}\hspace{3mm}(t\rightarrow\infty),
\end{equation}
so that
\begin{equation}
 \lim_{t\rightarrow\infty}\left(\frac{d}{dt}N_\lambda(t)\right)\leq\frac{2}{|\lambda|},
\end{equation}
which is exactly the $|\lambda|^{-1}$ suggested by (\ref{growthrates}).\\
To see why this theorem does not hold for all $\phi\notin\mathcal{D}(H)$, consider the Hilbert space $\mathcal{H}=L^2\left(\mathbb{R}_x\right)$ with the multiplication operator $H=x$ and the source state $\phi$ with $\phi(x)=\frac{1}{\sqrt{\pi(1+x^2)}}$.
In this case, the simple form of the overlap $\langle\phi,\mathrm{e}^{-iHt}\phi\rangle=e^{-|t|}$
allows for an explicit representation
\begin{equation}
\label{explicit}
\left(\mathrm{e}^{(-iH+\lambda P_\phi)t}\phi\right)(x)=e^{-ixt}\phi(x)\left(\lambda\frac{e^{(\lambda-1+ix)t}-1}{\lambda-1+ix}+1\right)
\end{equation}
for all $t\geq0$. As $t\rightarrow\infty$, for $\lambda<\tau=1$, $e^{iHt}\phi$ strongly converges to the limit function
\begin{equation}
 \phi(x)\cdot\frac{-1+ix}{\lambda-1+ix}
\end{equation}
with norm $1/\sqrt{1-\lambda}$, yielding the limit 
\begin{equation}
 \lim_{\lambda\rightarrow-\infty}\lim_{t\rightarrow\infty}\left(\frac{d}{dt}N_\lambda(t)\right)=\lim_{\lambda\rightarrow-\infty}\frac{2|\lambda|}{|1-\lambda|}=2,
\end{equation}
which is a positive saturation value for the particle flux. Also note that for the critical coupling constant $\lambda=1$,
\begin{equation}
 \left(\mathrm{e}^{(-iH+P_\phi)t}\phi\right)(x)=e^{-ixt}\phi(x)\left(1+\frac{e^{ixt}-1}{ix}\right)
\end{equation}
with
\begin{equation}
 \left\Vert\mathrm{e}^{(-iH+P_\phi)t}\phi\right\Vert^2=1+2t
\end{equation}
and thus $N_1(t)=t^2+t$. So in this case, the transition region implied by the theorems \ref{lingrowth} and \ref{pointspec} really consists only of the critical point $\lambda=\tau=1$, with quadratical growth of the particle number at the transition point itself, and exponential growth behaviour of the form $e^{(\lambda-1)t}$ for larger values of $\lambda$.

\section{Exponential growth}
\label{sectionexp}
\subsection{Proof of Theorem \ref{pointspec}}
\begin{proof}
It suffices to show that for any choice $H$, $\phi$ and $0<\epsilon<\frac{1}{2}$ the operator
\begin{equation}
 \frac{iH}{\lambda}+P_\phi
\end{equation}
has an eigenvalue inside the open ball $B_\epsilon(1)$ for sufficiently large $\lambda$. The idea of the proof is, that, given a closed operator $A$ on a Banach space $\mathcal{B}$, and a closed curve $\Gamma$ in the resolvent set of $A$ such that it separates the spectrum of $A$ in two parts $\sigma_1$ inside and $\sigma_2$ outside the curve, the operator defined by
\begin{equation}
 P=P_{A,\Gamma}=\frac{1}{2\pi i}\int_\Gamma(z-A)^{-1}\mathrm{d}z
\end{equation}
is a (non-orthogonal) projection with the properties 
\begin{equation}
\begin{split}
 PA&\subset AP\hspace{12mm}\mbox{with}\hspace{5mm}\sigma\left( A|P\mathcal{B}\right)=\sigma_1,\\
(\mathbf{1}-P)A&\subset A(\mathbf{1}-P)\hspace{5mm}\mbox{with}\hspace{5mm}\sigma\left( A|(\mathbf{1}-P)\mathcal{B}\right)=\sigma_2.
\end{split}
\end{equation}
Now an application of the Neumann series shows that if $B$ is a perturbation of $A$ with $\|B(A-z)^{-1}\|<1$ for all $z\in \Gamma$, one has
\begin{equation}
 \operatorname{dim}(P_{A,\Gamma}\mathcal{B})=\operatorname{dim}(P_{A+B,\Gamma}\mathcal{B}).
\end{equation}
In particular, if the spectrum of $A$ enclosed by $\Gamma$, $\sigma_1(A)$, consists of one eigenvalue of single multiplicity, then so does $\sigma_1(A+B)$. The details are given in Chapter IV, \S3 of \cite{kato}.\\
\indent{}Setting $\Gamma=\partial B_\epsilon(1)$, we apply this theorem twice, first with $\mathcal{B}=\mathcal{H}_E$, a suitable ``low energy'' subspace of $\mathcal{H}$, and then with $\mathcal{B}=\mathcal{H}$ to remove the energy cutoff again.\\
\indent{}Since $H$ is selfadjoint, one can define the spectral projection $Q_E=1_{[-E,E]}(H)$ by the functional calculus for any $E\geq0$, and write $\phi_E=Q_E\phi$. Since $Q_E\rightarrow\mathbf{1}_\mathcal{H}$ strongly as $E\rightarrow\infty$, there exists an $E>0$ with $\big\Vert\phi-\phi_E\big\Vert<\epsilon/8$. For this choice of $E$, take $\mathcal{H}_E=Q_E\mathcal{H}$ and $\lambda>{4E}/{\epsilon}$. Writing $P_{\phi_E}=|\phi_E\left\rangle \right\langle \phi_E |$, we first consider the operator
\begin{equation}
\label{ecop}
 \frac{iH}{\lambda}+P_{\phi_E}
\end{equation}
restricted to $\mathcal{H}_E$. $P_{\phi_E}$ has exactly one eigenvalue in $B_\epsilon(1)$, which is $\Vert\phi_E\Vert^2\in[1-\epsilon/4,1]$, and thus we have the estimate
\begin{equation}
 \max_{z\in\Gamma}\big\Vert iH\left(P_{\phi_E}-z\right)^{-1}/\lambda\big\Vert\leq\frac{E}{\lambda}\cdot\frac{4}{3\epsilon}<\frac{1}{3}.
\end{equation}
Therefore, the above-mentioned perturbation result applies, and the operator from (\ref{ecop}) has one single eigenvalue in $B_\epsilon(1)$. Furthermore, one can estimate its resolvent by the Neumann series to obtain
\begin{equation}
\label{resolv1}
\begin{split}
 \max_{z\in\Gamma}\big\Vert\left(iH/\lambda+P_{\phi_E}-z\right)^{-1}\big\Vert&\leq\max_{z\in\Gamma}\big\Vert\left( \mathbf{1}_{\mathcal{H}_E}+iH\left(P_{\phi_E}-z\right)^{-1}/\lambda\right)^{-1}\big\Vert\\&\hspace{6mm}\cdot\big\Vert\left(P_{\phi_E}-z\right)^{-1}\big\Vert\\&<(1-1/3)^{-1}\cdot4/(3\epsilon)={2}/{\epsilon}
\end{split}
\end{equation}
on the restricted space $\mathcal{H}_E$. On the orthogonal complement $\mathcal{H}_E^\perp$, we have ${iH}/{\lambda}+P_{\phi_E}={iH}/{\lambda}$, so this operator has purely imaginary spectrum, and its resolvent can easily be estimated by
\begin{equation}
 \max_{z\in\Gamma}\big\Vert\left(iH/\lambda+P_{\phi_E}-z\right)^{-1}\big\Vert=\max_{z\in\Gamma}\big\Vert\left(iH/\lambda-z\right)^{-1}\big\Vert\leq\frac{1}{1-\epsilon}<2.
\end{equation}
Since the operator ${iH}/{\lambda}+P_{\phi_E}$ on $\mathcal{H}$ is decomposed by the pair $\mathcal{H}_E,\mathcal{H}_E^\perp$ as described in \cite{kato}, Chapter III., \S5.6, one can simply combine the results concerning the subspaces, to obtain that the operator on the whole space still has one eigenvalue in $B_\epsilon(1)$ and obeys the estimate 
\begin{equation}
 \max_{z\in\Gamma}\big\Vert\left(iH/\lambda+P_{\phi_E}-z\right)^{-1}\big\Vert<{2}/{\epsilon}.
\end{equation}
Now we are ready for the second application of the perturbation result. As $\big\Vert P_\phi-P_{\phi_E}\big\Vert\leq2\Vert\phi-\phi_E\Vert<{\epsilon}/{4}$, we have
\begin{equation}
 \max_{z\in\Gamma}\big\Vert\left( P_\phi-P_{\phi_E}\right) \left(iH/\lambda+P_{\phi_E}-z\right)^{-1}\big\Vert<\frac{\epsilon}{4}\cdot\frac{2}{\epsilon}=\frac{1}{2}.
\end{equation}
As a consequence, ${iH}/{\lambda}+P_{\phi_E}+\left(P_\phi-P_{\phi_E}\right)= {iH}/{\lambda}+P_\phi$ has one eigenvalue of mulitiplicity one in $B_\epsilon(1)$.
\end{proof}

\subsection{Proof of Theorem \ref{alllambda}}
\begin{proof}
 The polynomial
\begin{equation}
 p(z):=z^4+2iz^3+(\lambda i-2)z^2-(2\lambda+i)z-(3/2)\lambda i 
\end{equation}
has a root with positive imaginary part for all $\lambda>0$. To show this let $z_k$, $k=1,...,4$ be its roots, counting mulitplicities. For all $z$ with $p(z)\neq0$, the product rule implies
\begin{equation}
 \frac{\overline{p(z)}p'(z)}{|p(z)|^2}=\frac{p'(z)}{p(z)}=\sum_{k=1}^4\frac{1}{z-z_k}
\end{equation}
 so if all $\operatorname{Im}z_k\leq0$, the right side, and therefore also $\overline{p(z)}p'(z)$ would never have a positive imaginary part for $z\in\mathbb{R}$. But
\begin{equation}
 \operatorname{Im}\left( \overline{p(\sqrt{\lambda/2})}p'(\sqrt{\lambda/2})\right) =\lambda/2+o(\lambda), \hspace{1cm}(\lambda\rightarrow0),
\end{equation}
 so $p$ must have a root with positive imaginary part for small positive $\lambda$. Furthermore, $p$ does not have a real root for any positive $\lambda$, because one can easily check that the real and imaginary part of 
\begin{equation}
 p(t)\cdot(t^3-2it^2-2t+i)=(t^6+1)\cdot t+3\lambda/2+\lambda(2it^5+it^3+2it)/2
\end{equation}
cannot equal zero for the same $t\in\mathbb{R}$. Since the roots of $p$ depend continuously on $\lambda$, this means that there is a $z_0$ with $\operatorname{Im}z_0>0$ and $p(z_0)=0$ for all $\lambda>0$.\\
The vector $\psi\in\mathcal{D}(H_0)$ which we will prove to be a (non-normalized) eigenvector is 
\begin{equation}
 \psi:=\left( H_0-{z_0^2}\right)^{-1}\phi. 
\end{equation}
Taking $C_R=\{|z|=R, \operatorname{Im}z\geq0\}\cup(-R,R)$ as contour of integration, one has by Cauchy's Integral theorem and by the fact that $\operatorname{Im}z_0>0$
\begin{equation}
\begin{split}
 \left\langle \phi, \psi\right\rangle &=\frac{3}{2\pi}\int_{\mathbb{R}}\frac{1}{(k^6+1)\left( {k^2}-{z_0^2}\right)}\mathrm{d}k\\&=\lim_{R\rightarrow\infty}\frac{3}{2\pi}\int_{C_R}\frac{1}{(z^6+1)\left( {z^2}-{z_0^2}\right)}\mathrm{d}z\\
&=i\frac{iz_0^2-2z_0-3i/2}{z_0^4+2iz_0^3-2z_0^2-iz_0}\\&=-\frac{i}{\lambda},
\end{split}
\end{equation}
where we have used $p(z_0)=0$ in the last line.
Therefore, $\psi$ is an eigenvector:
\begin{equation}
 \left( iH_0+\lambda P_\phi\right) \psi=iH_0\left( H_0-\frac{z_0^2}{2}\right)^{-1}\phi+\lambda\left\langle \phi, \psi\right\rangle\phi=i\phi+i\frac{z_0^2}{2}\psi-i\phi=\alpha(\lambda)\psi.
\end{equation}
Since $H_0$ has no eigenvectors, we have $P_\phi\psi\neq0$, and thus
\begin{equation}
 \|\psi\|^2\operatorname{Re}\alpha(\lambda)=\operatorname{Re}\left\langle\psi, \left( iH_0+\lambda P_\phi\right) \psi\right\rangle =\lambda|\left\langle\phi,\psi \right\rangle|^2>0.
\end{equation}
\end{proof}

\section{Restricted limits} 
\label{sectionrestrict}
\subsection{Independence (proof of Theorem \ref{independence})}
For the proof of Theorem \ref{independence}, we start with two lemmas:
\begin{lem}
\label{weak}
 For $\lambda<0$, $\phi$ an arbitrary normalized source state, 
\begin{equation}
 \mathrm{e}^{(-iH_0+\lambda P_\phi)t}\rightarrow0
\end{equation}
in the weak operator topology as $t\rightarrow\infty$. 
\end{lem}

\begin{proof}
 Let $\psi,\chi\in\mathcal{H}$. We have to show that $\left\langle\psi, \mathrm{e}^{(-iH_0+\lambda P_\phi)t}\chi\right\rangle\rightarrow0$ as $t\rightarrow\infty$. To do so, we write the scalar product as
\begin{equation}
\label{finally}
 \left\langle\psi, \mathrm{e}^{(-iH_0+\lambda P_\phi)t}\chi\right\rangle=\left\langle\psi, \mathrm{e}^{-iH_0t}\chi\right\rangle+\lambda\int_0^t\left\langle\psi,\mathrm{e}^{iH_0(t-s)}\phi\right\rangle\left\langle\phi, \mathrm{e}^{(-iH_0+\lambda P_\phi)s}\chi\right\rangle \mathrm{d}s
\end{equation}
The first term on the right-hand side tends to $0$ as $t\rightarrow\infty$ by the Riemann-Lebesgue Lemma (cf. \cite{perry}, p.\ 18, Example 1.1, note that $\mathcal{H}_{\mathrm{ac}}(H_0)=\mathcal{H}$, i.e. $H_0$ has only absolutely continuous spectrum). The second term will only be analyzed for all $\psi$ from a dense subset of $\mathcal{H}$. Since $\left\langle\psi, \mathrm{e}^{(-iH_0+\lambda P_\phi)t}\chi\right\rangle$ depends continuously on $\psi$ uniformly in $t>0$, this is enough to prove the assertion.\\
\indent{}According to Theorem 1.3 in \cite{perry} (p.\ 20), there is a subset $\mathcal{K}(H_0)$ which is dense in $\mathcal{H}_{a.c.}(H_0)=\mathcal{H}$ such that
\begin{equation}
 \left\langle\psi,\mathrm{e}^{iH_0t}\phi\right\rangle\in L^2(\mathbb{R}_t)
\end{equation}
for all $\psi\in \mathcal{K}(H_0)$. 
Thus, by (\ref{L2est}), the integral on the right side of (\ref{finally}) can be estimated by the convolution of two $L^2(\mathbb{R})$ functions, and therefore tends to $0$ as $t\rightarrow\infty$.
\end{proof}
\begin{lem}
\label{strong}
 For a set $\Omega\subset\mathbb{R}^d$ with finite Lebesgue measure $\mu(\Omega)<\infty$, and $\lambda<0$
\begin{equation}
 P_\Omega \mathrm{e}^{(-iH_0+\lambda P_\phi)t}\rightarrow0
\end{equation}
in strong operator topology as $t\rightarrow\infty$.
\end{lem}

\begin{proof}
 If $q>\frac{d}{2}, q\geq2$, one has $1_{\Omega}(x)\in L^q(\mathbb{R}^d_x)$ and $\left( \frac{p^2}{2}+i\right) ^{-1}\in L^q(\mathbb{R}^d_p)$, and by Theorem XI.20 in \cite{resi3}, (p.\ 47),
\begin{equation}
 P_\Omega\left(H_0+i \right)^{-1}= 1_{\Omega}(x)\left( \frac{(-i\nabla)^2}{2}+i\right) ^{-1}\in\mathcal{J}_q=\left\lbrace A\in\mathcal{B}(\mathcal{H}):\operatorname{tr}\left(|A|^q \right)<\infty. \right\rbrace
\end{equation}
Thus $ P_\Omega\left(H_0+i \right)^{-1}$ is compact. First, take $\psi\in\mathcal{D}(H_0)=H^2(\mathbb{R}^d)$ and apply the fact that a strongly continuous group commutes with its generator,
\begin{equation}
\begin{split}
 P_\Omega \mathrm{e}^{(-iH_0+\lambda P_\phi)t}\psi=P_\Omega&(H_0+i)^{-1}(H_0+i\lambda P_\phi+i-i\lambda P_\phi) \mathrm{e}^{(-iH_0+\lambda P_\phi)t}\psi\\
=P_\Omega&(H_0+i)^{-1}\mathrm{e}^{(-iH_0+\lambda P_\phi)t}(H_0+i\lambda P_\phi+i)\psi\\&-i\lambda P_\Omega(H_0+i)^{-1}P_\phi \mathrm{e}^{(-iH_0+\lambda P_\phi)t}\psi.
\end{split}
\end{equation}
Both terms on the right-hand side tend to $0$ in $\Vert\cdot\Vert_\mathcal{H}$ as $t\rightarrow\infty$. The first one because of Lemma \ref{weak} and the fact that the compact operator $ P_\Omega\left(H_0+i \right)^{-1}$ maps weakly convergent sequences to convergent ones, and the second one directly by Lemma \ref{weak}. Since $\mathcal{D}(H_0)$ is dense in $\mathcal{H}$, this already implies the claim by continuity.
\end{proof}
Now one can easily show that the initial density $\rho_0$ does not contribute to the value of $\lim_{t\rightarrow\infty}\operatorname{tr}(P_\Omega\rho_{_{F}}(t)P_\Omega)$ in the fermionic case.
\begin{proof}\textit{(Theorem \ref{independence})}
 Since $\mathrm{s}-\lim_{t\rightarrow\infty} P_\Omega \mathrm{e}^{(-iH_0+\lambda P_\phi)t}=0$ and $\rho_0\in\mathcal{J}_1$ (the trace class), Lemma 3.1 from \cite{nier} applies, and one has:
\begin{equation}
 \Vert P_\Omega \mathrm{e}^{(-iH_0+\lambda P_\phi)t}\rho_0\mathrm{e}^{(iH_0+\lambda P_\phi)t} P_\Omega\Vert_{\operatorname{tr}}
\leq\Vert P_\Omega \mathrm{e}^{(-iH_0+\lambda P_\phi)t}\rho_0\Vert_{\operatorname{tr}}\rightarrow0
\end{equation}
as $t\rightarrow\infty$.
\end{proof}

\subsection{Existence (proof of Theorems \ref{finlim}, \ref{fermicomp})}
The  main idea for the proofs of Theorems \ref{finlim} and \ref{fermicomp} is the formula
\begin{equation}
\label{timeder}
 \frac{d}{dt}\operatorname{tr}(P_\Omega\rho(t)P_\Omega)=2|\lambda|\Vert P_\Omega \mathrm{e}^{(-iH_0+\lambda P_\phi)t}\phi\Vert^2,
\end{equation}
which is obtained from (\ref{formal}) when setting $\rho_0=0$.
\begin{proof}\textit{(Theorem \ref{finlim})}
 The assumption of sublinear particle production (which is always the case for fermions) implies by (\ref{trg}) and (\ref{mono}), that 
\begin{equation}
\label{L^2}
 \int_0^\infty\vert\left\langle\phi, \mathrm{e}^{(-iH_0+\lambda P_\phi)s}\phi\right\rangle\vert^2\mathrm{d}s<\infty.
\end{equation}
Furthermore,
\begin{equation}
\label{triangle}
 \begin{split}
  \Vert P_\Omega &\mathrm{e}^{(-iH_0+\lambda P_\phi)t}\phi\Vert=\left\| P_\Omega \mathrm{e}^{-iH_0t}\phi+\lambda \int_0^t  P_\Omega \mathrm{e}^{-iH_0(t-s)}\phi \left\langle \phi ,\mathrm{e}^{(-iH_0+\lambda P_\phi)s}\phi\right\rangle \mathrm{d}s\right\|\\
&\leq\Vert P_\Omega \mathrm{e}^{-iH_0t}\phi\Vert+\vert\lambda\vert \int_0^t \Vert P_\Omega \mathrm{e}^{-iH_0(t-s)}\phi \Vert\vert\left\langle \phi ,\mathrm{e}^{(-iH_0+\lambda P_\phi)s}\phi\right\rangle\vert \mathrm{d}s
 \end{split}
\end{equation}
By assumption (\ref{L1}), $\Vert P_\Omega \mathrm{e}^{-iH_0t}\phi\Vert\in L^1\cap L^\infty(\mathbb{R}^+_t)$. Thus, the first term in (\ref{triangle}) obviously is in $L^2(\mathbb{R}^+_t)$, and by (\ref{L^2}), the second one is a convolution of an $L^1(\mathbb{R}^+)$ function with an $L^2(\mathbb{R}^+)$ function, and therefore also contained in $L^2(\mathbb{R}^+_t)$. By (\ref{timeder}), this proves the assertion.
\end{proof}
\begin{proof}\textit{(Theorem \ref{fermicomp})}
 Let $K:=\operatorname{supp}(\hat{\phi})$ be the support of $\hat{\phi}$. Since Fourier transform diagonalizes $H_0$,

\begin{equation}
 \mathrm{e}^{(-iH_0+\lambda P_\phi)t}\phi=\mathrm{e}^{-iH_0t}+\lambda \mathrm{e}^{-iH_0t}\int_0^t \mathrm{e}^{iH_0s}\phi \left\langle \phi, \mathrm{e}^{(-iH_0+\lambda P_\phi)s}\phi \right\rangle \mathrm{d}s
\end{equation}
implies $\operatorname{supp}\left(\mathcal{F} \mathrm{e}^{(-iH_0+\lambda P_\phi)t}\phi\right)\subset K$ for all times $t\geq0$. This means
\begin{equation}
 P_\Omega \mathrm{e}^{(-iH_0+\lambda P_\phi)t}\phi=P_\Omega 1_K(-i\nabla) \mathrm{e}^{(-iH_0+\lambda P_\phi)t}\phi,
\end{equation}
 where $P_\Omega 1_K(-i\nabla)$ is a Hilbert-Schmidt operator with corresponding norm
\begin{equation}
 \Vert P_\Omega 1_K(-i\nabla)\Vert_{\mathcal{J}_2}\leq(2\pi)^{-d/2}\left( |\Omega||K|\right) ^{1/2}
\end{equation}
by theorem XI.20 in \cite{resi3}, p.\ 47. Applying (\ref{L2est}) to a singular value decomposition of $P_\Omega 1_K(-i\nabla)$ one obtains
\begin{equation}
\begin{split}
  \lim_{t\rightarrow\infty}\operatorname{tr}(P_\Omega\rho_{_{F}}(t)P_\Omega)&=2|\lambda|\int_0^\infty\Vert P_\Omega \mathrm{e}^{(-iH_0+\lambda P_\phi)t}\phi\Vert^2\mathrm{d}t\\
&=2|\lambda|\int_0^\infty\Vert P_\Omega 1_K(-i\nabla) \mathrm{e}^{(-iH_0+\lambda P_\phi)t}\phi\Vert^2\mathrm{d}t\\
&\leq\Vert P_\Omega 1_K(-i\nabla)\Vert_{\mathcal{J}_2}^2
\\&\leq(2\pi)^{-d}|\Omega||K|.
\end{split}
\end{equation}
\end{proof}

\section{Semiclassical limit (proof of Theorem \ref{limitdist})}
\label{sectionsemiclass}
To obtain a semiclassical limit, it is important to observe that the distinction between the bosonic and fermionic character of the semigroups disappears in the limit $\epsilon\rightarrow0$.

\begin{lem}
\label{osc}
 For any $\phi,\psi\in\mathcal{H}$, $\Vert\phi\Vert=1$ one has for all $t,c\in\mathbb{R}$:
\begin{equation}
\lim_{\epsilon\rightarrow0}\Vert \exp({-iH_0{t}/{\epsilon}+cP_\phi t})\psi-\exp(-iH_0{t}/{\epsilon})\psi\Vert=0.
\end{equation}
\end{lem}

\begin{proof}
One can write this difference as
\begin{equation}
\label{firstintegral}
\begin{split}
 \exp({-iH_0{t}/{\epsilon}+cP_\phi t})&\psi-\exp(-iH_0{t}/{\epsilon})\psi\\&\hspace{-10mm}=c\int_0^t\left\langle \phi,\exp({-iH_0s/\epsilon{}+cP_\phi s})\psi\right\rangle \exp(-iH_0(t-s)/{\epsilon})\phi \mathrm{d}s.
\end{split}
\end{equation}
For the scalar product in the integral of (\ref{firstintegral}), one has
\begin{equation}
\label{secondintegral}
\begin{split}
 &\left\langle \phi,\exp({-iH_0s/\epsilon{}+cP_\phi s})\psi\right\rangle=\left\langle \phi,\exp({-iH_0s/\epsilon{}})\psi\right\rangle\\&\hspace{10mm}+c\int_0^s\left\langle \phi,\exp({-iH_0{r}/{\epsilon}+cP_\phi r})\psi\right\rangle \left\langle \phi,\exp(-iH_0({s-r})/{\epsilon})\phi\right\rangle \mathrm{d}r.
\end{split}
\end{equation}
The first term in (\ref{secondintegral}) tends to zero for all $s\neq0$ as $\epsilon\rightarrow0$ by the Riemann-Lebesgue Lemma. By the same argument, the integrand of the second term tends to zero almost everywhere, and it is bounded by ${e}^{\vert cs\vert}$. By dominated convergence this implies
\begin{equation}
 \lim_{\epsilon\rightarrow0}\left\langle \phi,\exp({-iH_0s/\epsilon{}+cP_\phi s})\psi\right\rangle=0.
\end{equation}
for all $s\neq0$. But dominated convergence also applies to (\ref{firstintegral}) and the assertion is proven.
\end{proof}

To show Theorem \ref{limitdist}, it is useful to have the inverse for the Wigner transform, which is given by the Weyl quantization. For a function $a\in L^2\left(\mathbb{R}^d_x\times\mathbb{R}^d_p \right)$ one can define an operator on $\mathcal{H}$ by
\begin{equation}
 \label{quantization}
(\operatorname{Op}[a]\psi)(x)=\frac{1}{(2\pi)^d}\int_{\mathbb{R}^{d}_p}\int_{\mathbb{R}^{d}_y}a(\frac{x+y}{2},p)\mathrm{e}^{i p\cdot(x-y)}\psi(y) \mathrm{d}y \hspace{0.5mm}\mathrm{d}p
\end{equation}
 Up to a factor, this is a unitary map from $L^2\left(\mathbb{R}^d\times\mathbb{R}^d \right)$ to the space $\mathcal{J}_2$ of Hilbert-Schmidt operators on $\mathcal{H}$ (cf. equation (3.3) in \cite{nier}):
\begin{equation}
 \label{L2J2}
\int_{\mathbb{R}^{2d}}\overline{a(x,p)}b(x,p)\mathrm{d}x \hspace{0.5mm}\mathrm{d}p=(2\pi)^d\operatorname{tr}\left(\operatorname{Op}[a]^* \operatorname{Op}[b]\right).
\end{equation}
With this definition, one has for all Hilbert-Schmidt operators $\kappa$,
\begin{equation}
 \operatorname{Op}\left[ W[\kappa]\right] =\frac{1}{(2\pi)^d}\kappa,
\end{equation}
where we identified operators and kernels. This is the main tool for the proof of Theorem \nolinebreak \ref{limitdist}.
\begin{proof}
  Let $\theta\in\mathcal{D}\left(\mathbb{R}^{d}_X\times\mathbb{R}^{d}_P\right)=C_0^\infty\left(\mathbb{R}^{d}_X\times\mathbb{R}^{d}_P\right)$ be a test function. We have to establish the existence of 
\begin{equation}
 \lim_{\epsilon\rightarrow0}\int_{\mathbb{R}^{2d}}f^\epsilon(X,P,T)\theta(X,P)\mathrm{d}X\hspace{0.5mm}\mathrm{d}P.
\end{equation}
By (\ref{L2J2}) one can transform this integral to
\begin{equation}
\begin{split}
 \int_{\mathbb{R}^{2d}}f^\epsilon(X,P,T)\theta(X,P)\mathrm{d}X \hspace{0.5mm}\mathrm{d}P&=\int_{\mathbb{R}^{2d}}\epsilon^{-d}W\left[ \rho^\epsilon\left(T/\epsilon\right)\right] \left({X}/\epsilon,P \right)\theta(X,P)\mathrm{d}X \hspace{0.5mm}\mathrm{d}P\\
&=\int_{\mathbb{R}^{2d}}W\left[ \rho^\epsilon\left(T/\epsilon\right)\right] \left(x,p \right)\theta(\epsilon x,p)\mathrm{d}x \hspace{0.5mm}\mathrm{d}p\\
&=\operatorname{tr}\left( \rho^\epsilon\left(T/\epsilon\right)   \operatorname{Op}[\theta^\epsilon]\right),
\end{split}
\end{equation}
where $\operatorname{Op}[\overline{a}]=\operatorname{Op}[a]^*$ is used, and $\theta^\epsilon(x,p)=\theta(\epsilon x, p)$.\par
By (\ref{formal}), the following expression holds,
\begin{equation}
\label{bosonsol}
 \begin{split}
  \rho^\epsilon\left( T/\epsilon{}\right)= \sgn(c)&\left( \exp({-iH_0T/\epsilon{}+cP_\phi T}) \exp({iH_0T/\epsilon{}+cP_\phi T})-\mathbf{1}\right) \\&+\exp({-iH_0T/\epsilon{}+cP_\phi T})\rho^\epsilon_0 \exp({iH_0T/\epsilon{}+cP_\phi T}).\hspace{4mm}
\end{split}
\end{equation}
To keep notation simple, we first assume $\rho^\epsilon_0=0$.\\
\indent{}The estimate $\Vert\left( \vert \psi\left\rangle \right\langle \psi\vert - \vert \chi\left\rangle \right\langle \chi\vert\right)\Vert_{\operatorname{tr}}\leq\Vert\psi-\chi\Vert\cdot\left( \Vert\psi\Vert+\Vert\chi\Vert\right) $ yields, together with Lemma \ref{osc} and the dominated convergence theorem,
\begin{equation}
\label{J1conv}
 \begin{split}
\Vert\rho^\epsilon\left( T/\epsilon{}\right)&-2|c|\int_0^T\mathrm{e}^{-iH_0S/\epsilon{}} P_\phi \mathrm{e}^{iH_0S/\epsilon{}}\mathrm{d}S
\Vert_{\operatorname{tr}}\\&
\leq2|c|\int_0^T\|\mathrm{e}^{-iH_0S/\epsilon{}+cP_\phi S}P_\phi \mathrm{e}^{iH_0S/\epsilon{}+cP_\phi S}-\mathrm{e}^{-iH_0S/\epsilon{}} P_\phi \mathrm{e}^{+iH_0S/\epsilon{}}\|_{\operatorname{tr}}\mathrm{d}S\\
&\leq2|c|\left( {e}^{|cT|}+1\right) \int_0^T\Vert \mathrm{e}^{-iH_0S/\epsilon{}+cP_\phi S}\phi-\mathrm{e}^{-iH_0S/\epsilon{}}\phi\Vert \mathrm{d}S\rightarrow0\hspace{4mm}(\epsilon\rightarrow0)
 \end{split}
\end{equation}
Furthermore, by the theorem of Calder\'{o}n-Vaillancourt (Theorem 2.8.1 in \cite{martinez}), the operators $\operatorname{Op}[\theta^\epsilon]$ are uniformly bounded in $\mathcal{B}(\mathcal{H})$ as $\epsilon$ tends to zero, and thus, by the linearity and cyclicity of the trace,
\begin{equation}
 \lim_{\epsilon\rightarrow0}\left|\operatorname{tr}\left(\rho^\epsilon\left( T/\epsilon{}\right)  \operatorname{Op}[\theta^\epsilon]\right)-2|c|\int_0^T\operatorname{tr}\left(\mathrm{e}^{iH_0S/\epsilon{}}\operatorname{Op}[\theta^\epsilon]\mathrm{e}^{-iH_0S/\epsilon{}} P_\phi \right) \mathrm{d}S\right|=0.
\end{equation}
The Heisenberg time evolution applied to $\operatorname{Op}[\theta^\epsilon]$ can be carried over to phase space by defining 
\begin{equation}
\label{defeta}
 \eta^{\epsilon}_S(x,p):=\theta^\epsilon\left( x+pS/\epsilon{},p\right)=\theta(\epsilon x+pS,p)
\end{equation}
so that (cf. \cite{nier})
\begin{equation}
\label{heisenberg}
 \mathrm{e}^{iH_0S/\epsilon{}}\operatorname{Op}[\theta^\epsilon]\mathrm{e}^{-iH_0S/\epsilon{}}=\operatorname{Op}[\eta^{\epsilon}_S].
\end{equation}
Now Lemma 3.2 from \cite{nier} applies to $\eta^{\epsilon}_S$, stating that the operators ${\operatorname{Op}[\eta^{\epsilon}_S]}_{\epsilon\geq0}$ are uniformly bounded in $\mathcal{B}(\mathcal{H})$ and converge strongly to ${\operatorname{Op}[\eta^{0}_S]}$ as $\epsilon$ tends to zero. Thus also $\operatorname{tr}\left(\operatorname{Op}[\eta^\epsilon_S]P_\phi \right)$ converges and, by dominated convergence, the desired limit exists,
\begin{equation}
\label{heureka}
\begin{split}
 \lim_{\epsilon\rightarrow0}\int_{\mathbb{R}^{2d}}f^\epsilon(X,P,T)\theta(X,P)\mathrm{d}X\hspace{0.5mm}\mathrm{d}P&=\lim_{\epsilon\rightarrow0}2|c|\int_0^T\operatorname{tr}\left(\operatorname{Op}[\eta^\epsilon_S] P_\phi \right)\mathrm{d}S\\
&=2|c|\int_0^T\operatorname{tr}\left(\operatorname{Op}[\eta^0_S] P_\phi \right)\mathrm{d}S.
\end{split}
\end{equation}
By (\ref{defeta}), $\eta^0_S$ does not depend on $X$. $\operatorname{Op}[\eta^0_S]$ is a multiplication operator in momentum space and thus
\begin{equation}
\label{value0}
\begin{split}
 2|c|\int_0^T\operatorname{tr}\left(\operatorname{Op}[\eta^0_S] P_\phi \right) \mathrm{d}S&=2|c|\int_0^T\int_{\mathbb{R}^d}\theta(PS,P)\vert\hat\phi(P)\vert^2\hspace{0.5mm}\mathrm{d}S.\\
\end{split}
\end{equation}
At this point of the proof, one should remark that $\eta^0_S$ is no longer in $L^2\left(\mathbb{R}^d\times\mathbb{R}^d \right)$ so $\operatorname{Op}[\eta^0_s]$ is not given by our definition of the Weyl quantization (\ref{quantization}). But it can be obtained as the quantization of a $C^\infty\left(\mathbb{R}^d\times\mathbb{R}^d\right)$ symbol with all derivatives bounded. This construction is described, for example, in Chapter 2 of \cite{martinez}.\par
Next, we allow a general initial condition, i.e.\ $\rho^\epsilon_0$ is an arbitrary trace-norm bounded sequence of density matrices such that (\ref{indist}) holds for some distribution $g$ on macroscopic phase space. Concentrating on the last term in (\ref{bosonsol}), one has to evaluate
\begin{equation}
 \operatorname{tr}\left(\mathrm{e}^{-iH_0T/\epsilon{}+cP_\phi T}\rho^\epsilon_0 \mathrm{e}^{iH_0T/\epsilon{}+cP_\phi T}\operatorname{Op}[\theta^\epsilon]\right).
\end{equation}
As before, one first has to check that the action of the semigroups can be substituted by the free evolution,
\begin{equation}
 \begin{split}
  &\Vert \mathrm{e}^{-iH_0T/\epsilon{}+cP_\phi T}\rho^\epsilon_0 \mathrm{e}^{iH_0T/\epsilon{}+cP_\phi T}\operatorname{Op}[\theta^\epsilon]-\mathrm{e}^{-iH_0T/\epsilon{}}\rho^\epsilon_0 \mathrm{e}^{iH_0T/\epsilon{}}\operatorname{Op}[\theta^\epsilon]\Vert_{\operatorname{tr}}\\
&\leq \Vert \mathrm{e}^{iH_0T/\epsilon{}}\mathrm{e}^{-iH_0T/\epsilon{}+cP_\phi T}\rho^\epsilon_0 \mathrm{e}^{iH_0T/\epsilon{}+cP_\phi T}\mathrm{e}^{-iH_0T/\epsilon{}}-\rho^\epsilon_0\Vert_{\operatorname{tr}}\cdot\Vert \operatorname{Op}[\theta^\epsilon]\Vert_{\mathcal{B}\left( \mathcal{H}\right) }\\
&\leq 2|c| {e}^{|cT|}\int_0^T\Vert P_\phi \mathrm{e}^{-iH_0S/\epsilon{}+cP_\phi S}\rho^\epsilon_0 \Vert_{\operatorname{tr}}\mathrm{d}S   \cdot\Vert\operatorname{Op}[\theta^\epsilon]\Vert_{\mathcal{B}\left( \mathcal{H}\right) }.
 \end{split}
\end{equation}
The operator norm in the last line is bounded uniformly by the Calder\'{o}n - Vaillancourt theorem. For the integral, Lemma \ref{osc} and the uniform boundedness of $\|\rho_0^\epsilon\|_{\operatorname{tr}}$ allow an application of dominated convergence, 
\begin{equation}
\begin{split}
 \limsup_{\epsilon\rightarrow0}\int_0^T\Vert P_\phi \mathrm{e}^{-iH_0S/\epsilon{}+cP_\phi S}\rho^\epsilon_0 \Vert_{\operatorname{tr}}\hspace{0.5mm}\mathrm{d}S&=\limsup_{\epsilon\rightarrow0}\int_0^T\Vert P_\phi \mathrm{e}^{-iH_0S/\epsilon{}}\rho^\epsilon_0 \Vert_{\operatorname{tr}}\hspace{0.5mm}\mathrm{d}S\\
&=\limsup_{\epsilon\rightarrow0}\int_0^T\Vert \rho_0^\epsilon \mathrm{e}^{iH_0S/\epsilon{}}\phi\Vert \hspace{0.5mm}\mathrm{d}S=0.
\end{split}
\end{equation}
The last equality is seen as follows: Since the spectrum of $H_0$ is absolutely continuous, we first can consider a state $\phi$ with bounded spectral density, i.e. 
\begin{equation}
 \frac{d\mu_{\phi}(E)}{dE}\leq C^2
\end{equation}
for some $C>0$. Then for any $\psi\in\mathcal{H}$, $g(t)=\left\langle\psi, \mathrm{e}^{iH_0t}\phi \right\rangle$ is the Fourier transform of an $L^2(\mathbb{R})$ function, satisfying
\begin{equation}
 \Vert g\Vert_{L^2}\leq\sqrt{2\pi}C\Vert\psi\Vert.
\end{equation}
Therefore, the Cauchy-Schwarz inequality yields
\begin{equation}
\label{estimate}
 \int_0^T\big|\left\langle\psi, \mathrm{e}^{iH_0S/\epsilon}\phi \right\rangle\big|\hspace{0.5mm}\mathrm{d}S\leq\sqrt{2\pi T}C\Vert\psi\Vert\sqrt{\epsilon}.
\end{equation}
Now we can write 
\begin{equation}
 \rho^\epsilon_0=\sum_{n\in\mathbb{N}}a^\epsilon_n|\psi^\epsilon_n\left\rangle \right\langle \psi^\epsilon_n|
\end{equation}
with $\left( \psi^\epsilon_n\right) $ an $\epsilon$-dependent orthonormal basis and $\sum_{n\in\mathbb{N}}|a^\epsilon_n|\leq K$ uniformly in $\epsilon$. Then (\ref{estimate}) implies
\begin{equation}
\begin{split}
 \int_0^T\Vert \rho_0^\epsilon \mathrm{e}^{iH_0S/\epsilon{}}\phi\Vert \hspace{0.5mm}\mathrm{d}S&\leq\sum_{n\in\mathbb{N}}|a^\epsilon_n|\int_0^T\big|\left\langle\psi^\epsilon_n, \mathrm{e}^{iH_0S/\epsilon}\phi \right\rangle\big|\hspace{0.5mm}\mathrm{d}S\\
&\leq K\sqrt{2\pi T}C\sqrt{\epsilon}\rightarrow0\hspace{0.5cm}(\epsilon\rightarrow0).
\end{split}
\end{equation}
Approximating general $\phi$ with states $\phi_C$ of bounded spectral density, one can use the uniform boundedness of the operators $\rho_0^\epsilon \mathrm{e}^{iH_0S/\epsilon{}}$ to show
\begin{equation}
 \limsup_{\epsilon\rightarrow0}\int_0^T\Vert \rho_0^\epsilon \mathrm{e}^{iH_0S/\epsilon{}}\phi\Vert \hspace{0.5mm}\mathrm{d}S\leq TK\Vert\phi-\phi_C\Vert,
\end{equation}
where the right hand side can be made arbitrarily small. Thus one can in fact replace the semigroups by the unitary group, and then apply (\ref{L2J2}),
\begin{equation}
\label{valueinit}
\begin{split}
 \lim_{\epsilon\rightarrow0}&\operatorname{tr}\left(\mathrm{e}^{-iH_0T/\epsilon{}+cP_\phi T}\rho^\epsilon_0 \mathrm{e}^{iH_0T/\epsilon{}+cP_\phi T}\operatorname{Op}[\theta^\epsilon]\right)\\
&=\lim_{\epsilon\rightarrow0}\operatorname{tr}\left(\mathrm{e}^{-iH_0T/\epsilon{}}\rho^\epsilon_0 \mathrm{e}^{iH_0T/\epsilon{}}\operatorname{Op}[\theta^\epsilon]\right)\\
&=\lim_{\epsilon\rightarrow0}\operatorname{tr}\left(\rho^\epsilon_0 \operatorname{Op}[\eta^\epsilon_T]\right)\\
&=\lim_{\epsilon\rightarrow0}\int_{\mathbb{R}^{2d}}W\left[ \rho^\epsilon_0\right] \left(x,p \right)\theta(\epsilon x+pT, p)\mathrm{d}x\hspace{0.5mm}\mathrm{d}p\\
&=\lim_{\epsilon\rightarrow0}\int_{\mathbb{R}^{2d}}\epsilon^{-d}W\left[ \rho^\epsilon_0\right] \left(X/\epsilon,P \right)\theta(X+PT,P)\mathrm{d}X\hspace{0.5mm}\mathrm{d}P\\
&=\int_{\mathbb{R}^{2d}}g(X,P)\theta(X+PT, P)\mathrm{d}X\hspace{0.5mm}\mathrm{d}P\\
&=\int_{\mathbb{R}^{2d}}g(X-PT,P)\theta(X, P)\mathrm{d}X\hspace{0.5mm}\mathrm{d}P.
\end{split}
\end{equation}
Adding (\ref{value0}) and (\ref{valueinit}), one obtains the limit for general initial conditions,
\begin{equation}
\label{value}
\begin{split}
 \lim_{\epsilon\rightarrow0}&\int_{\mathbb{R}^{2d}}f^\epsilon(X,P,T)\theta(X,P)\mathrm{d}X\hspace{0.5mm}\mathrm{d}P\\&=\int_{\mathbb{R}^{2d}}g(X-PT,P)\theta(X, P)\mathrm{d}X\hspace{0.5mm}\mathrm{d}P\\&\hspace{2cm}+2|c|\int_0^T\int_{\mathbb{R}^d}\theta(PS,P)\vert\hat\phi(P)\vert^2\mathrm{d}P\hspace{0.5mm}\mathrm{d}S\\
&=\left\langle f^0(T),\theta\right\rangle_{\mathcal{D}', \mathcal{D}}
\end{split}
\end{equation}
The distribution in the last line is given by 
\begin{equation}
f^0(X,P,T)=g(X-PT,P)+2|c|\int_0^T\delta(X-PS)\vert\hat\phi(P)\vert^2\mathrm{d}S.
\end{equation}
\noindent{}It remains to show (\ref{limitdiff}). By (\ref{indist}), the initial value is
\begin{equation}
 f^0(X,P,0)=\lim_{\epsilon\rightarrow0}\epsilon^{-d}W\left[ \rho^\epsilon_0\right] \left(X/\epsilon,P \right)= g(X,P).
\end{equation}
By testing with an arbitrary $\theta\in\mathcal{D}\left(\mathbb{R}^{d}_X\times\mathbb{R}^{d}_P\right)$, one can check that $f^0$ also solves the differential equation,
\begin{equation}
 \begin{split}
  &\frac{d}{dT}\left\langle f^0(T),\theta\right\rangle_{\mathcal{D}', \mathcal{D}}=\\
&=\frac{d}{dT}\left(\int_{\mathbb{R}^{2d}}g(X,P)\theta(X+PT, P)\mathrm{d}X\hspace{0.5mm}\mathrm{d}P\right)\\&\hspace{2cm}+\frac{d}{dT}\left(2|c|\int_0^T\int_{\mathbb{R}^d}\theta(P(T-S),P)\vert\hat\phi(P)\vert^2\mathrm{d}P\hspace{0.5mm}\mathrm{d}S \right)\\ 
&=\int_{\mathbb{R}^{2d}}g(X,P)P\cdot\nabla_X\theta(X+PT, P)\mathrm{d}X\hspace{0.5mm}\mathrm{d}P\\&\hspace{1cm}+2|c|\int_0^T\int_{\mathbb{R}^d}P\cdot\nabla_X\theta(P(T-S),P)\vert\hat\phi(P)\vert^2\mathrm{d}P\hspace{0.5mm}\mathrm{d}S\\&\hspace{2cm}+2|c|\int_{\mathbb{R}^d}\theta(0,P)\vert\hat\phi(P)\vert^2dP\\
&=\left\langle f^0(T),P\cdot\nabla_X\theta\right\rangle_{\mathcal{D}', \mathcal{D}}+\left\langle 2|c|\delta(X)\vert\hat\phi(P)\vert^2,\theta\right\rangle_{\mathcal{D}', \mathcal{D}}\\
&=\left\langle -P\cdot\nabla_Xf^0(T)+2|c|\delta(X)\vert\hat\phi(P)\vert^2,\theta\right\rangle_{\mathcal{D}', \mathcal{D}}.
 \end{split}
\end{equation}
\end{proof}

\vspace{5mm}
\textit{Acknowledgements.} We would like to thank V. S. Buslaev and A. Komech for instructive and encouraging discussions. 

M. Butz acknowledges support from the ENB graduate program TopMath and a grant sponsored by Max Weber-Programm Bayern.

\end{document}